\documentclass{article}
\usepackage{amsthm}
\usepackage{amsmath}
\usepackage{amssymb}

\newtheorem{thm}{Theorem}

\newcommand{\skl}[1]{\left\langle #1\right\rangle}
\newcommand{\fanin}[1]{\mathrm{fanin}(#1)}
\newcommand{\fanout}[1]{\mathrm{fanout}(#1)}
\newcommand{\x}{\vec{x}}
\newcommand{\y}{\vec{y}}

\newcommand{\e}{\vec{e}}
\renewcommand{\a}{\vec{a}}
\renewcommand{\u}{\vec{u}}

\newcommand{\nulis}{\vec{0}}
\newcommand{\gf}{\{0,1\}} 
\newcommand{\size}[1]{s(#1)} 
\newcommand{\sized}[2]{s_{#2}(#1)} 

\newcommand{\image}[1]{\mathrm{Im}(#1)} 
\newcommand{\oper}[1]{\mu(#1)} 

\newcommand{\STOC}[1]{Proc. #1 STOC}
\newcommand{\MFCS}[1]{Proc. #1 MFCS}

\begin{document}

\title{Circuits with Arbitrary Gates for Random Operators
\thanks{Research of both authors supported by a DFG grant SCHN~503/4-1.
University of Frankfurt, Institute of Computer Science,
Frankfurt am Main, Germany}}

\author{S.~Jukna \and G.~Schnitger}

\date{}

\maketitle

\begin{abstract}
We consider boolean circuits computing $n$-operators 
$f:\{0,1\}^n\to \{0,1\}^n$.
As gates we allow arbitrary boolean functions; neither fanin nor fanout 
of gates is restricted.
An operator is linear if it computes $n$ linear forms, that is, 
computes a matrix-vector product $A\x$ over $GF(2)$.

We prove the existence of
$n$-operators requiring about $n^2$ wires in any circuit, and
linear $n$-operators requiring about $n^2/\log n$ wires in depth-$2$ 
circuits, if either all output gates or all gates on the middle layer 
are linear.  
\end{abstract}

\section{Introduction}

We consider general circuits computing $n$-operators
$f:\{0,1\}^n\to\{0,1\}^n$. As gates we allow \emph{arbitrary} boolean
functions of their inputs; there is no restriction on their fanin or
fanout.  Thus, the phenomenon which causes complexity of such circuits
is \emph{information transfer} rather than \emph{information
  processing} as in the case of single functions.  Such a circuit is a
directed acyclic graph with $n$ input nodes $x_1,\ldots,x_n$ and $n$
output nodes $y_1,\ldots,y_n$. Each non-input node computes some
boolean function of its predecessors. A circuit computes an operator
$f=(f_1,\ldots,f_n)$ if, for
all $i=1,\ldots,n$, the boolean function computed at the $i$th output
node $y_i$ is the $i$th component $f_i$ of the
operator~$f$.  The \emph{depth} of a circuit is the
largest number of wires in a path from an input to an output node.

The \emph{size} of a circuit is the total number of wires in it.  We
will denote by $\sized{f}{d}$ the smallest number of wires in a
general circuit of depth at most $d$ computing $f$. If there are no
restrictions on the depth, the corresponding measure is denoted by
$\size{f}$.  Note that $\size{f}\leq \sized{f}{1}\leq n^2$ holds for
any $n$-operator, so quadratic lower bounds are the highest ones.

Circuits of depth $2$ constitute the first non-trivial model.
Interest in depth-$2$ circuits comes from the following important
result of Valiant~\cite{valiant}: If in every depth-$2$ circuit,
computing $f$ with $O(n/\ln\ln n)$ gates on the middle layer, at least
$n^{1+\Omega(1)}$ direct wires must connect inputs with output gates,
then $f$ cannot be computed by log-depth circuits with a linear number
of fanin-$2$ gates.  To prove a super-linear lower bound for log-depth
circuits is an old and well-known problem in circuit complexity.

Super-linear lower bounds up to $\sized{f}{2}=\Omega(n\log^2 n)$ where
proved using graph-theoretic arguments by analyzing some
super-concentration properties of the circuit as a graph
\cite{DDPW,Pippenger1,Pippenger2,PS,Pudlak,AP,PRS,RT,RS}.  Higher
lower bounds of the form $\sized{f}{2}=\Omega(n^{3/2})$ were recently
proved using information theoretical arguments~\cite{cherukhin,juk}.
For larger depth $d$ known lower bounds are only slightly non-linear.
All these bounds, however, are on the \emph{total} number of wires, so
they still have no consequences for log-depth circuits.

In fact, in the class of general circuits, even the question about the
complexity of a \emph{random} operator remained unclear.  In
particular, it was unclear whether operators requiring a quadratic
number of wires (even in depth $2$) exist at all?

\section{Circuits for general operators}

Note that a direct counting argument, as in the case of constant fanin
circuits, does not work for general circuits: already for $d>n+\log
n$, the number $2^{2^{d}}$ of possible boolean functions that may be
assigned to a node of fanin $d$ may be larger than the total number
$2^{n2^n}$ of $n$-operators.

Our first result is an observation that this bad situation can be
excluded by just turning the power of circuits against themselves to
ensure that, in an optimal circuit, no gate can have fanin larger
than~$n$.  This leads us to

\begin{thm}\label{thm:1}
  For almost all $n$-operators $f$, $\size{f}=\Omega(n^2)$.
\end{thm}

\begin{proof}
  Let $\oper{L}$ be the number of different $n$-operators computable
  by boolean circuits with at most $L$ wires. Our goal is to upper
  bound this number in terms of $n$ and $L$, and compare this bound
  with the total number $2^{n2^n}$ of $n$-operators.

  Take an optimal circuit with $\ell\leq L$ wires computing some
  $n$-operator; hence, $\ell\leq n^2$.  Then $\ell=\sum_{i=1}^m d_i$,
  where $d_1,\ldots,d_m$ are the fanins of its gates.  It is clear
  that we need $m\geq n$ gates, since we must have $n$ input gates. On
  the other hand, $m\leq \ell+n+2\leq 2n^2$ gates are always enough
  since every non-input gate, besides two possible constant gates,
  must have nonzero fanin.

  We now make use of the fact that the gates in our circuits may be
  \emph{arbitrary} boolean functions: This allows us to assume that
  $d_i\leq n$ for all $i$.  Indeed, if $d_i>n$, then we can replace
  the $i$th gate by the boolean function computed at this gate and
  join it to all $n$ input variables; when doing this, the total
  number of wires in the circuit can only decrease.

  The number of sequences $d_1,\ldots,d_m$ of fanins with $0\leq
  d_i\leq n$ does not exceed $(n+1)^m$. For each such sequence and for
  each $i=1,\ldots,m$, there are at most ${{m}\choose{d_i}}\leq
  m^{d_i}$ possibilities to chose the set of inputs for the $i$th node
  and at most $2^{2^{d_i}}$ possibilities to assign a boolean function
  to this node. Hence,
  \[
  \oper{L}\leq (n+1)^m\prod_{i=1}^mm^{d_i} \prod_{i=1}^m2^{2^{d_i}}
  =(n+1)^m m^{\sum_{i=1}^m d_i}2^{\sum_{i=1}^m2^{d_i}}\,.
  \]
  Since $\sum_{i=1}^m d_i\leq L\leq n^2$ and $m\leq 2n^2$, this yields
  \[
  \log_2 \oper{L}\leq \sum_{i=1}^m 2^{d_i}+O(n^2\log_2 n)\,.
  \]
  We now observe that at most $n/2$ nodes can have fanin larger than
  $2L/n$, for otherwise we would have more than $(2L/n)\cdot(n/2)=L$
  wires in total. Since $m\leq 2n^2$ and since the fanin of each gate
  does not exceed $n$, we obtain that
  \[
  \sum_{i=1}^m 2^{d_i}\leq (m-n/2)2^{2L/n}+(n/2)2^n \leq
  2n^24^{L/n}+n2^{n-1}\,.
  \]
  Hence,
  \begin{equation}\label{eq:1}
    \log_2 \oper{L}\leq 2n^24^{L/n}+n2^{n-1}+ O(n^2\log_2 n)\, .
  \end{equation}
  Since the total number of operators $f:\{0,1\}^n\to\{0,1\}^n$ is
  $2^{n2^n}$, the smallest number $L$ of wires sufficient to compute
  all of them must satisfy $\log_2\oper{L}\geq n2^n$.  By
  (\ref{eq:1}), this implies
  \[
  2n^24^{L/n}\geq n2^{n-1}-O(n^2\log_2 n)\,.
  \]
  Dividing both sides by $2n^2$, we obtain that $4^{L/n}=
  \Omega(2^n/n)$, and hence, $L=\Omega(n^2)$.
\end{proof}

\section{Circuits for linear operators}

An important class of operators are \emph{linear} ones.  Each such
operator computes $n$ linear forms, that is, computes a matrix-vector product
$f_A(\x)=A\x$ over $GF(2)$ where $A$ is an $n\times n$ $(0,1)$-matrix.
We are interested in the complexity $\sized{f_A}{2}$ of such operators
in the class of depth-$2$ circuits.

If all gates are required to be \emph{linear} (parities and their
negations), then easy counting shows that some linear operators
require $\Omega(n^2/\log n)$ wires.  It is also known that $O(n^2/\log
n)$ are also sufficient to compute any linear
operator~\cite{tuza,bublitz,AKW}.

But what if we allow arbitrary (non-linear) boolean functions as
gates---can we then compute linear operators $f_A$ more efficiently?
The largest known lower bound for an \emph{explicit} linear operator
$f_A$ has the form $\sized{f_A}{2}=\Omega(n\log n)$~\cite{Pudlak}.
This raises the following question: Do \emph{linear} $n$-operators
requiring $\sized{f_A}{2}=\Omega(n^2/\log n)$ wires exist at all?  We
are only able to answer this question \emph{positively} under the
additional restriction that either all output gates of all gates on
the middle layer must be linear functions.

The next theorem shows that the non-linearity of \emph{middle} gates is no
problem: any such circuit can be transformed into a \emph{linear}
circuit with almost the same number of wires.  Hence, some linear
$n$-operators require about $n^2/\log n$ wires in such circuits.

\begin{thm}\label{thm:2}
  If a depth-$2$ circuit computes a linear $n$-operator and only has linear
  gates on the output layer, then it can be transformed to an
  equivalent linear circuit by adding at most $2n$ new wires.
\end{thm}

\begin{proof}
  Let $A$ be an $n$-by-$n$ $(0,1)$-matrix, and let $\Phi$ be a
  depth-$2$ circuit computing $A\x$. We may assume, for simplicity,
  that there are no direct wires from inputs to outputs: this can be
  easily achieved by adding $n$ new wires on the first level. Assume
  that all output gates of $\Phi$ are linear boolean functions.  By
  adding one constant-$1$ function on the middle layer and at most $n$
  new wires on the second level, we can also assume that each output
  gate computes just the sum modulo $2$ of its inputs (and not the
  negation of this sum).

  Let $h=(h_1,\ldots,h_r):\gf^n\to\gf^r$ be the operator computed by
  the gates on the middle layer. Since $A\nulis=\nulis$ and each
  output gate computes the sum modulo $2$ of its inputs, we may assume
  that $h(\nulis)=\nulis$ as well: If $h_j(\nulis)=1$ for some $j$,
  then replace the function $h_j$ by the function $h_j'$ such that
  $h_j'(\nulis)=0$ and $h_j'(\x)=h_j(\x)$ for all $\x\neq\nulis$.

  Let $B$ be the $n$-by-$r$ adjacency $(0,1)$-matrix of the bipartite
  graph formed by the wires joining the gates on the middle layer with
  those on the output layer.  Then $ A\x=B\cdot h(\x) $ for all
  $\x\in\gf^n$.  Write each vector $\x=(x_1,\ldots,x_n)$ as the linear
  combination $\x=\sum_{i=1}^n x_i\e_i$ of unit vectors
  $\e_1,\ldots,\e_n\in\gf^n$, and replace the operator $h$ computed on
  the middle layer by a \emph{linear} operator
  \[
  h'(\x):=\sum_{i=1}^n x_ih(\e_i) \bmod 2\,.
  \] Hence, $h'(\x)=\x^{\top}M$, where $M$ is an $n\times r$ matrix
  with rows $h(\e_1),\ldots,h(\e_n)$.  Using the linearity of the
  matrix-vector product, we obtain that (with all sums mod~$2$):
  \begin{align*}
    B\cdot h(\x)&= A\cdot \Big(\sum x_i\e_i\Big) =\sum x_i A\e_i =\sum
    x_i B\cdot h(\e_i) =B\cdot h'(\x)\,.
  \end{align*}
  Hence, the new (linear) circuit $\Phi'$ computes $A\x$ as well. It
  remains to show that the number of wires in $\Phi'$ does not exceed
  the number of wires in $\Phi$.

  The wires on the second level haven't changed at all.  To show that
  the \emph{number} of wires on the first level has not increased as
  well, let $\fanout{x_i}$ be the fanout of the $i$th input node
  $x_i$, and $\fanin{h_j}$ the fanin of the $j$th gate $h_j$ on the
  middle layer.  Then $\sum_{i=1}^n
  \fanout{x_i}=\sum_{j=1}^r\fanin{h_j}$ is the total number $L$ of
  wires on the first level.  We know that $h(\nulis)=\nulis$, that is,
  $h_j(\nulis)=0$ for all $j=1,\ldots,r$.  Now we make a simple (but
  crucial) observation: if there is no wire from $x_i$ to $h_j$, then
  $h_j(\e_i)=h_j(\nulis)=0$.  This implies that the $j$th column of
  $M$ can have at most $\fanin{h_j}$ ones. Since the number of wires
  on the first level of $\Phi'$ is just the total number of $1$'s in
  $M$, we are done.
\end{proof}

The second case---when only gates on the middle layer are required to
be linear---is more delicate.  That such circuits \emph{can} be more
powerful than linear ones, was shown in \cite{juk1}. Given a boolean
$n\times n$ matrix $A$, say that a circuit \emph{weakly computes} the
operator $f_A(\x)=A\x$ if it correctly computes it on all $n$ unit
vectors $\e_1,\ldots,\e_n$.  Note that, for \emph{linear} circuits,
this is no relaxation: such a circuit weakly computes $f_A$ iff it
correctly computes $f_A$ on all inputs.  Hence, some linear operators
cannot be weakly computed by \emph{linear} depth-$2$ circuits using
fewer than $\Omega(n^2/\log n)$ wires.  It is however shown in
\cite{juk1} that the situation changes drastically if we only use
linear gates on the middle layer but allow
non-linear gates on the output layer, then \emph{any} linear
$n$-operator can be weakly computed using only $O(n\log n)$ wires.

Still, using Kolmogorov complexity arguments, we can prove that,
for some matrices $A$, such circuits require a quadratic number of wires
to compute the entire operator~$A\x$.

\begin{thm}\label{thm:3}
  If middle gates are required to be linear, then linear $n$-operators
  $f_A$ with $\sized{f_A}{2}=\Omega(n^2/\log n)$ exist.
\end{thm}

\begin{proof}
  We use the Kolmogorov complexity argument known as the
  \emph{incompressibility argument} (see \cite{LV} for background).
  Since we have $2^{n^2}$ matrices, some matrix $A$ requires $n^2$
  bits to describe it. Hence, the linear operator $f_A(\x)=A\x$ cannot
  be described using fewer than $n^2-O(1)$ bits, as well.

  Fix an arbitrary depth-$2$ circuit $\Phi$ computing $f_A$, and
  assume that all its gates on the middle layer are linear.  Let $L$
  be the number of wires in~$\Phi$. As before, we may assume that
  there are no direct wires from inputs to outputs.  Our goal is to
  show that, using the circuit $\Phi$, the operator $f_A$ can be
  described using $O(L\log n)$ bits. This will imply the desired lower
  bound $L=\Omega(n^2/\log n)$ on the number of wires.

  Let $r$ be the number of nodes on the middle layer of $\Phi$.  Since
  at these nodes only linear functions are computed, the first level
  (between inputs and middle layer) computes some linear operator
  $\y=B\x$, where $B$ is the $r$-by-$n$ adjacency matrix of the
  bipartite graph formed by the wires joining the gates on the input
  layer with those on the middle layer.  Let also $C$ be the
  $n$-by-$r$ adjacency matrix of the bipartite graph formed by the
  wires joining the gates on the middle layer with those on the output
  layer.  Hence, $L=|B|+|C|$ where $|B|$ denotes the number of $1$s
  in~$B$.

  Using these two matrices $B$ and $C$ as well as the fact that the
  operator computed by the circuit $\Phi$ is linear, we can encode
  this operator using $O(L\log n)$ bits as follows.

  \begin{enumerate}
  \item[$\circ$] Since $|B|+|C|=L$, both matrices $B$ and $C$ can be
    described using $O(L\log n)$ bits, just by describing the
    positions of their $1$-entries.

  \item[$\circ$] The $i$th output gate of $\Phi$ computes $g_i(B\x)$,
    where $g_i:\{0,1\}^r\to\{0,1\}$ is some boolean function depending
    only on rows of $B$ seen by this gate, that is, on rows
    corresponding to the $d_i$ nodes on the middle layer seen by this
    gate. Let $B_i$ be the $d_i\times n$ submatrix of $B$ formed by
    these rows.

    Let $\image{B_i}=\{B_i\x\colon \x\in \{0,1\}^n\}$ be the column
    space of $B_i$.  If this space has dimension $t$ then any $t$
    linearly independent columns of $B$ form its basis.  Take the set
    $B_i'=\{\u_1,\ldots,\u_{t}\}$ of the first $t$ linearly
    independent columns of $B_i$, and call it the \emph{first basis}
    of $\image{B_i}$.

  \item[$\circ$] Encode the behavior of $g_i$ on this basis $B_i'$ by
    the string $g_i(\u_1),\ldots,g_i(\u_t)$ of $t\leq d_i$ bits. The
    entire string, for all $n$ output gates $g_1,\ldots,g_n$, has
    length at most $\sum_{i=1}^n d_i\leq L$.
  \end{enumerate}

  Having this encoding, we can recover the value $g_i(\x)$ of the
  $i$th output gate on a given input $\x\in \{0,1\}^n$ as follows.

  \begin{enumerate}
  \item Compute $\y_i=B_i\x$. We can do this since the $i$th row of
    $C$ tells us what rows of $B$ appear in $B_i$, and we know the
    entire matrix $B$.

  \item Take the first basis $B_i'$ of $\image{B_i}$ and write $\y_i$
    as a linear combination $\y_i=\sum_{k=1}^t\lambda_k \u_k$ of basis
    vectors over $GF(2)$.

  \item Give $z_i=\sum_{k=1}^t\lambda_k g_i(\u_k)\bmod 2$ as an
    output. We can compute this number since we know the values
    $g_i(\u_1),\ldots,g_i(\u_{t})$.
  \end{enumerate}

  Since the circuit computes $A\x$, the $i$th output gate must compute
  the scalar product $\skl{\a_i,\x}$ of input vector $\x$ with the
  $i$th row $\a_i$ of $A$.  Hence, $g_i(B\x)=\skl{\a_i,\x}$, meaning
  that $g_i$ must be \emph{linear} on $\image{B}$. Since $g_i$ can
  only see the middle gates corresponding to the rows of $B_i$, this
  implies that $g_i$ must be linear also on $\image{B_i}$. Thus,
  \[
  z_i=\sum_{k=1}^t\lambda_k g_i(\u_k)=g_i\Big(\sum_{k=1}^t\lambda_k
  \u_k \Big) =g_i(\y_i)=g_i(B_i\x)=g_i(B\x)\,,
  \]
  that is, $z_i$ is a scalar product of $\x$ with the $i$th row of
  $A$, as desired.
\end{proof}

\section{Conclusion}

We have shown that, even when arbitrary boolean functions can be used
as gates, some operators $f:\{0,1\}^n\to\{0,1\}^n$ require about $n^2$
wires.  We have also shown that some \emph{linear} operators require
about $n^2/\log n$ wires in depth-$2$ circuits, if either all output
gates or all gates on the middle layer are required to be linear.

We conjecture that the same lower bound for depth-$2$ circuits
computing linear operators should also hold without any restrictions
on used gates.

\end{document}